%% file: main_version.tex
\pdfoutput=1
\documentclass[11pt,twoside]{article}

\usepackage{xcolor}

\usepackage{fullpage}

\usepackage{epsf}
\usepackage{fancyhdr}
\usepackage{graphics}
\usepackage{graphicx} 
\usepackage{float} 
\usepackage{subfigure} 
\usepackage{psfrag}
\usepackage{comment}
\def\P{\mbox{P}}

\usepackage[linesnumbered,ruled]{algorithm2e}
\DontPrintSemicolon	

\usepackage{color}
\usepackage{amsthm}
\usepackage{amsfonts}
\usepackage{amsmath}
\usepackage{bm}
\usepackage{amssymb,bbm}
\usepackage[numbers]{natbib}
\usepackage{algorithmic}
\usepackage[usestackEOL]{stackengine}

\usepackage{url}
\usepackage[colorlinks=True,linkcolor=magenta,citecolor=blue,urlcolor=blue,pagebackref=true,backref=page]
{hyperref}
\renewcommand*{\backref}[1]{\ifx#1\relax \else Page #1 \fi}
\renewcommand*{\backrefalt}[4]{%
    \ifcase #1 \footnotesize{(Not cited.)}%
    \or        \footnotesize{(Cited on page~#2.)}%
    \else      \footnotesize{(Cited on pages~#2.)}%
    \fi}

\usepackage{nicefrac}

\def\N{\mbox{N}}

\usepackage{chngpage}

\usepackage{tabularx}%

\usepackage{enumitem}
\usepackage{booktabs}
\usepackage{pbox}

\usepackage{caption}

\usepackage{mathtools}

\usepackage{fullpage}
\allowdisplaybreaks
\input{final_macros.tex}

\SetKwInput{KwInput}{Input}                
\SetKwInput{KwOutput}{Output}              

\newcommand{\expectation}{\mathbb{E}}
\newcommand{\Var}{\text{Var}}
\newcommand{\expect}{\mathbb{E}}
\newcommand{\iidsim}{\overset{\mathrm{iid}}{\sim}}
\newcommand{\hattheta}{\hat{\theta}}

\theoremstyle{remark}
\newtheorem*{remark}{Remark}

\begin{document}

  
%




\title{A Bayesian Bootstrap for Mixture Models }
\author{Fuheng Cui \& Stephen G. Walker \\ \\
Department of Statistics and Data Sciences \\
The University of Texas at Austin\\
email: fuheng.cui@austin.utexas.edu, s.g.walker@math.utexas.edu
}

\date{}

\maketitle

\begin{abstract}

This paper proposes a new nonparametric Bayesian bootstrap for a mixture model, by developing the traditional Bayesian bootstrap. We first reinterpret the  Bayesian bootstrap, which uses the P\'olya-urn scheme, as a gradient ascent algorithm which associated one-step solver. The key then is to use the same basic mechanism as the Bayesian bootstrap with the switch from a point mass kernel to a continuous kernel. Just as the Bayesian bootstrap works solely from the empirical distribution function, so the new Bayesian bootstrap for mixture models works off the nonparametric maximum likelihood estimator for the mixing distribution. From a theoretical perspective, we prove the convergence and exchangeability of the sample sequences from the algorithm and  also illustrate our results with different models and settings and some real data.

\end{abstract}


Stochastic gradient algorithm.



\section{Introduction}
\label{sec:introduction}
The Bayesian bootstrap (BB), introduced in \cite{Rubin_1981}, is a data driven, prior free, nonparametric posterior, which is a particular version of the Dirichlet process; see \cite{Ferguson_1973}. The derivation is quite straightforward and only requires the empirical distribution function of an observed sample.  
While the presentation given by \cite{Rubin_1981} is in terms of random weights assigned to the observations, a more instructive consideration is the corresponding P\'olya-urn scheme, which constructs a probability model $p(x_{n+1:\infty}\mid x_{1:n})$; more accurately $p(x_{n+1:\infty}\mid F_n)$, where $F_n$ is the empirical distribution function of the observed sample $x_{1:n}$. The P\'olya-urn model is described in terms of balls in urns; i.e. sample a ball, replace it and add one more of the same color. Unfortunately, this hides a nice principle about the BB which is that the  $p(x_{m+1:\infty}\mid x_{1:m})$ for any $m\geq n$ is sequential where the $x_{m+1}$ is sampled from the current empirical distribution function of $x_{1:m}$, and then the empirical itself, i.e. the current distribution estimator, is updated using this newly acquired sample. Since one can only sample an observed data point, the weights get updated, and a martingale property ensures convergence to the random weights, which are Dirichlet and given in \cite{Rubin_1981}.
By now there are a number of different extension to the foundational Bayesian bootstrap, including for massive datasets, and for intractable likelihoods. See, for example, \cite{Barrientos_2020} and \cite{Vo_2019} for some recent applications including references to other types of Bayesian bootstrap.

The aim of the present paper is to use the principle of sampling from the current estimator and updating it with the new sample, to extend the BB idea to a nonparametric mixing distribution. Specifically, we demonstrate how to extend the BB from an empirical distribution function to a mixing distribution. The BB is well known to assign random weights to the points of the empirical distribution function and are distributed according to a Dirichlet distribution with common parameters set to 1; see \cite{Ferguson_1973} and \cite{Rubin_1981}. Indeed, there is the connection with the Dirichlet process where the parameter is $n\,F_n$. 

However, it is not clear how the random weighting of the atoms of the empirical distribution in the case of the BB could be extended and developed to a mixing distribution, from which no observations have been directly observed. The motivation for a BB approach for a mixing distribution is apparent. Current Bayesian nonparametric inference for a mixing distribution involve either complicated Markov chain Monte Carlo algorithms, or algorithms which make forms of approximations; see, for example, \cite{Gershman_2012} and \cite{Favaro_2013}. 

To proceed to this problem, we first reconsider the BB via the P\'olya-urn scheme, which also requires some degree of reinterpretation. We then rewrite the sequence as a novel stochastic gradient algorithm and it is this construct which allows us to see how to move the BB idea to a mixing distribution.  

The set up is as follows. For the BB we start with the empirical distribution function 
$$F_n(x)=n^{-1}\sum_{i=1}^n 1(x_i\leq x)$$
where the $(x_i)$ are an observed sample of size $n$, and for the measure version we would use the $1(x=y)$ indicator function,
\begin{equation}
    1(x=y) = \left\{
    \begin{aligned}
        1, \quad x=y,\\
        0, \quad x\neq1.
    \end{aligned}
    \right.
\end{equation}
A random distribution, or probability measure, taken from the BB, equivalently a data dependent, i.e. prior free, posterior is given by
$$F(x)=\sum_{i=1}^n w_i\,1(x=x_i)$$
where the $w$ is Dirichlet with common parameter 1, that is, $p(w)\propto 1\left(\sum_{i=1:n}w_i=1\right)$.

The random distribution can also be generated by the P\'olya-urn scheme. As we have indicated previously,  we write in the language of sampling from the empirical distribution rather than the taking of balls from urns. So $p(x_{n+1:\infty}\mid x_{1:n})$ arises by taking $x_{n+1}$ from $F_n$, then update the empirical distribution with the new sample to get
$$F_{n+1}(x)=\frac{n\,F_{n}(x)+1(x=x_{n+1})}{n+1}.$$
This process is repeated, so sample $x_{n+2}$ from $F_{n+1}$ and construct $F_{n+2}$ in the obvious way, and so on. It is seen the sampling from the empirical and updating using the new sample is equivalent to sampling a ball from an urn, replacing it and adding another of the same kind of ball. For us it is important to the see scheme as sampling and updating the empirical distribution function in order for us to extend the BB to a mixing distribution.

The limit $F_\infty$ exists (see for example \cite{Berti2006}) and is a random distribution function. This is easy to see since for each $x$ it is that $(F_m(x))_{m>n}$ forms a bounded martingale. This follows since
$$E\,[1(x=x_{m+1})\mid F_m]=F_m(x).$$
If we write the limiting distribution function as $F_\infty$ then it can be shown that the weights are coming from the aforementioned Dirichlet distribution. An easier view on this is to consider the sequence of random weights $(w_{i,m})_{i=1:n}$ where $w_{i,n}=1/n$. Then
$$w_{i,m+1}=\frac{m\,w_{i,m}+1(x_{m+1}=x_i)}{m+1},$$
and recall $P(x_{m+1}=x_i\mid w_m)=w_{i,m}$. It is now convenient to write this as
\begin{equation}\label{weights}
w_{i,m+1}=w_{i,m}+\alpha_{i,m}\left[\frac{1(x_{m+1}=x_i)}{w_{i,m}}-1\right],
\end{equation}
where $\alpha_{i,m}=w_{i,m}/(m+\alpha)$ and $\alpha=1$. This can be seen to be in the form of a stochastic gradient algorithm.

\subsection{BB as a stochastic gradient algorithm}

Stochastic gradient algorithms are, in the most simplest case, of the form 
$$v_{m+1}=v_m+\alpha_m\,s(v_m,x_{m+1})$$
for some stochastic sequence $(v_m)$ which arises due to the random sequence $(x_m)$. 

For these kind of algorithms, first note that the $(\alpha_{i,m})$ in (\ref{weights}) satisfy the usual conditions in~\cite{Robbins_1951}, namely
$$\sum_{m}\alpha_{i,m}=\infty\quad\mbox{and}\quad \sum_{m}\alpha_{i,m}^2<\infty\quad\mbox{for each}\quad i.$$
Further, the term in the square brackets of (\ref{weights}) has zero expectation, ensuring that the sequence $(w_{i,m})$ is a martingale for each $i=1:n$. We can see the term as the gradient of a particular objective function. To see this, consider
$$l(w)=\log\left( \sum_{j=1}^n w_{j,m} 1(x_{m+1}=x_j)\right)\,+\lambda\left(\sum_{j=1}^m w_{j,m}-1\right),$$
where the $l(w)$ is the log-likelihood and subject to the constraint that the sum of the weights is 1, and hence the Lagrange multiplier. Now
$$\frac{\partial l}{\partial w_i}=\frac{1(x_{m+1}=x_i)}{w_{i,m}}+\lambda.$$
Setting this to 0, we get $1(x_{m+1}=x_i)=-\lambda w_{i,m}$, so we see that we must take $\lambda=-1$. Hence, the term in square brackets in (\ref{weights}) is $\partial l/\partial w_i$. 
It is this algorithm and its motivation which we will develop for a mixing distribution.

The BB just described produces an exchangeable sequence for the $(x_{n+1}:\infty)$. This is a nice outcome, though we would, even if it was not such an outcome, sample from and update the empirical distribution function. The thinking is that we would use what we regard as the best distribution for $x_{m+1}$ given $x_{1:m}$ at any point in time. If we replaced the empirical distribution, which is discrete, and which can be considered as a drawback, by a continuous distribution, the sequence of $(x_{n+1:\infty})$ may not be exchangeable. However, if done in such a way so the sequence of updated distribution functions forms a martingale, then the sequence will be conditionally identically distributed (c.i.d.), a notion which relaxes exchangeability, and was originally studied by \cite{Berti_2004}. This was the theme of \cite{Fong_2021}. We aim to combine all these elements to present a BB for the mixing distribution for which the sequence of future unobserved $x_{n+1:\infty}$ are either c.i.d. or asymptotically exchangeable.  

\subsection{Nonparametric mixture model}
\label{sec:overview_of_nonparametric_estimates}

The mixture model has density function given by
\begin{align}
\label{def:mixture_model}
    f(y;G)=\int_{\Omega} k(y;\theta)\,dG(\theta),
\end{align}
where $k(y;\theta)$ is a known kernel with $\int k(y;\theta)\mathrm{d}y=1$, for all $\theta \in \Omega \subset \mathbb{R}^d$. To reinforce the difference between the model based on the BB we will now represent data as $y$ rather than $x$. Based on a sample of size $n$, the nonparametric maximum likelihood estimator (NPMLE) exists under the mild regularity condition that $k$ is bounded. We write this as $\widehat{G}$ and the existence and uniqueness and the discreteness is detailed in \cite{Lindsay_1983}. Indeed, $\widehat{G}$ is discrete with at most $n$ atoms.  

The likelihood function for $G$ is given by
\begin{align}
    \label{equ:npmle_problem}
    l(G)=\prod_{i=1}^{n} \int_{\Omega} k\left(y_{i}; \theta\right) \,\mathrm{d}G(\theta).
\end{align}
The full details for the existence and discreteness of the NPMLE is provided in the following:

\begin{theorem}[Existence and Discreteness of NPMLE]
    \label{thm:discrete_NPMLE}
    If $\Gamma=\{k(\cdot;\theta)\mid \theta\in\Omega\}$ is closed and bounded, then there exists a $\widehat{G}(\theta)$ which maximizes $l(G)$ in (\ref{equ:npmle_problem}), and $\widehat{G}$ can be written as
    $$
    \widehat{G}(\theta)=\sum_{j=1}^r \pi_j 1(\theta=\theta_j),
    $$
    where $\sum_{j=1}^r \pi_j = 1$ and  $r\leq n$.
\end{theorem}

There are a number of methods to compute the NPMLE, some of which can be easy and fast, see for example~\cite{Koenker_2014} and \cite{Chae_2018}. In \cite{Chae_2018}, the authors use Bayesian ideas to construct an iterative algorithm to find the NPMLE. Specifically, assume $G_0$ is a starting  distribution, then the update is given by the average of posteriors:
\begin{equation}
\label{equ:iterative_algorithm}
    G_{t+1}(\theta)=\frac{1}{n}\sum_{i=1}^n \frac{\int_{s\leq \theta} k(y_i;s)\,dG_t(s)}{\int_{\Omega}k(y_i;s)\mathrm{d}G_t(s)}.
\end{equation}
In \cite{Chae_2018} it is shown that the sequence $(G_t)$ converges to $\widehat{G}$.
Our aim is to use $\widehat{G}$ to generate the BBM, just like the $F_n$ generates the BB, and hence Bayesian uncertainty quantification for the mixing distribution. When the NPMLE is not available, we use alternative estimators, which will be detailed later in the paper.

Describing the layout of the paper: The bootstrap algorithms derived from the P\'olya-urn scheme and Bayesian Martingale scheme are shown in Section~\ref{sec:The Bayesian Bootstrap for Mixture Models (BBM)}. In Section~\ref{sec:convergence_and_exchangeability} we show some theoretical properties of our algorithm, including convergence and asymptotic exchangeability. Section~\ref{sec:illustration} presents some illustrations for some different models and settings. We conclude with a brief discussion and summary in Section~\ref{sec:summary_and_discussion}.

\section{The Bayesian bootstrap for mixture models (BBM)}
\label{sec:The Bayesian Bootstrap for Mixture Models (BBM)}

The bootstrap and the Bayesian bootstrap are fundamental tools for providing \emph{Uncertainty Quantification} (UQ)
about an estimator. Specifically, in the nonparametric case, uncertainty with respect to the empirical distribution function as an estimator of the true distribution. 

As we have seen with the BB, the basic idea is as follows. Start with an estimator of the distribution, say $\widehat{F}_n$. To provide UQ about this estimator, sample from $\widehat{F}_n$ to get $x_{n+1}$ and then update $\widehat{F}_n$ to $\widehat{F}_{n+1}$ using the new sample $x_{n+1}$.  

Our starting point is the P\'olya-urn model and we will study this in further detail to see how we can extend to other nonparametric models, specifically the mixture model.
The P\'olya-urn model is an original scheme in Bayesian nonparametric statistics, see \cite{Ghosal_2010}. 
We already note that it is connected with the Dirichlet process, see \cite{Blackwell_1973}, and is also related to other urn models, such as \cite{Hoppe_1984}. 

Applying the P\'olya-urn scheme, equivalently the BB, on an empirical distribution function $\widehat{F}_n$ has the outcome of randomly weighting the mass on each data point. 
The starting weights are equally $w_{0,i}=1/n$ for $i=1,\ldots,n$. If now the current weights at iteration $m$ are $w_{m,i}$ then the updated weights can be written as
\begin{equation}\label{Polya_urn}
w_{m+1,i}=w_{m,i}+\alpha_m\big(1(x_{m+1}=x_i)-w_{m,i}\big),
\end{equation}
for $m\geq 0$.
The term in brackets has expectation 0 since $\P(x_{m+1}=x_i\mid x_{1:m})=w_{m,i}$. Hence the $(w_{m,i})$ is a martingale sequence and hence converges for each $i$ to $w_{\infty,i}$.  

Specifically, even though more data are being generated, the number of atoms is fixed, it is only the weights which are being randomized.

Suppose now we move to the mixture model and the corresponding NPMLE, based on a  finite sample set $(y_1,\ldots,y_n)$, independently and identically distributed drawn from (\ref{def:mixture_model}). 
We write the NPMLE as
\begin{equation}
\label{equ:NPMLE_form}
    \widehat{G}(\theta)=\sum_{j=1}^r \pi_j 1(\theta=\theta_j)\quad \text{ with } \quad \sum_{i=1}^r \pi_i = 1,\quad\mbox{and} \quad r\leq n.
\end{equation}
The corresponding data density estimator is given by
$$\widehat{f}_n(y)=\sum_{j=1}^r \pi_j\, k(y\mid\theta_j).$$
One way to look to see how to proceed with an adaption of the BB is to see that the form for $\widehat{f}_n$ returns the empirical mass function when we exchange the kernel for a point mass at the $(\theta_j)$, which become the data points, and so $r=n$. 

To develop the BB for the mixture model, we would start by sampling $y_{n+1}$ from $\widehat{f}_n$. We then need to update $\widehat{f}_n$ to $\widehat{f}_{n+1}$ using $y_{n+1}$. The case is made that we only need to update the weights, i.e. the $(\pi_j)$, and the locations, i.e. the $(\theta_j)$, which is equivalent to an update of the current $\widehat{G}$. This is in keeping with the BB in that the weights get updated, and given the introduction of the kernel, so the parameters get updated, which do not exist with the BB. We argue that the number of atoms $r$ does not need to be updated, just as it is not updated either with the BB. All the new data are generated from the current kernels and so no new kernel is required. Hence, no new kernel location is required, so $r$ stays at it is. Therefore, we do not need new atoms to explain $(y_{n+1}:\infty)$. If we only update the weights and leave the locations fixed, then the sequence $(y_{n+1:\infty})$ is easily seen to be c.i.d.; whereas if we update both weights and locations then, as we shall prove later, the sequence is asymptotically exchangeable. 

To see how to update the weights and the kernel parameters, we set up a stochastic gradient algorithm for the mixture model.
To this end we consider the one-step log-likelihood for the weights with the constraint; i.e.
\begin{equation}\label{equ:polya_opt_constraint}
    \max_{w,\theta}\,  \log\left[\sum_{j=1}^r w_{j} k(y_{n+1}\mid \theta_j)\right]
    \quad \mbox{subject to}\quad  \sum_{j=1}^r w_{j}=1,
\end{equation}
using for now the $y_{n+1}$ data point which we have already described as to how it is obtained. Using Lagrange multipliers the optimization function becomes
\begin{align}
    \label{equ:log-likelihood_location_update}
    l(w,\theta) = \log \left[\sum_{j=1}^r w_j k(y_{n+1},\theta_j)\right]-\left(\sum_{j=1}^r w_j-1\right),
\end{align}
where $w=(w_1,\ldots,w_r)$ and $\theta=(\theta_1,\ldots,\theta_r)$. The derivatives of interest here are
\begin{align*}
    \frac{\partial l}{\partial w_j}=\frac{k(y_{n+1}\mid\theta_j)}{\sum_{i=1}^r w_i\,k(y_{n+1}\mid\theta_i)}-1,\quad
    \nabla_{\theta_j}l = \frac{w_j}{\sum_{i=1}^r w_j\,k(y_{n+1}\mid\theta_i)} \nabla_{\theta}k(y_{n+1}\mid\theta_i),
\end{align*}
for all $j\in\{1,\ldots,r\}$. 

\begin{algorithm}[ht!]
\DontPrintSemicolon
  \KwInput{Data collected $\{y_1,\ldots, y_n\}$, the estimate $G_0(\theta)=\sum_{j=1}^r w_{j,0}1(\theta=\theta_{j,0})$, the step sizes $\alpha_{j, m}=w_{j, m}/(m+n)$, $\beta_{j, m}= g\left(y_{n+m+1}, w_{m}, \theta_{m}\right)/(m+n)$ and $M$ is the number of posterior samples.}
  \For{$k= 0, 1,2,\ldots,M$}{
  \For{$m=0,1,\ldots$}{
  Sample $y_{m+n+1}\sim p_{m}(y)=\int k(y\mid\theta)\, dG_{m}(\theta)$;\\
  Update $w_{j,m+1}=w_{j, m}+\alpha_{j, m+1}\left(\frac{k\left(y_{n+m+1}, \theta_{j, m}\right)}{\sum_{i=1}^{r} w_{i, m} k\left(y_{n+m+1}, \theta_{i, m}\right)}-1\right);$\\
  Update $\theta_{j, m+1}=\theta_{j, m}+\beta_{j, m+1}\left(\frac{w_{j, m}}{\sum_{i=1}^{r} w_{i, m} k\left(y_{n+m+1}, \theta_{i, m}\right)} \nabla_{\theta_{j}} k\left(y_{n+m+1}, \theta_{j, m}\right)\right);$\\
  Set $G_{m+1}=\sum_{j=1}^r w_{j,m+1}1(\theta=\theta_{j,m+1}).$
  }
  Set $\hattheta_k=\theta_{\infty}$, $\hat{w}_k=w_{\infty}$
  }
  \KwOutput{$\{(\hattheta _k,\hat{w}_k)\}_1^M$}
\caption{The Bayesian Bootstrap for Mixture Models (BBM)}
\label{algo:The Bayesian Bootstrap for Mixture Models (BBM)}
\end{algorithm}

So the update rule is
\begin{align}
w_{j,m}&= w_{j,m-1}+\alpha_{j,m}\left[\frac{k(y_m,\theta_{j,m-1})}{\sum_{i=1}^r w_{i,m-1} k(y_m,\theta_{i,m-1})}-1\right],\label{equ:update_rule_weights}\\
\theta_{j,m}&= \theta_{j,m-1}+\beta_{j,m}\frac{w_{j,m-1}\nabla_{\theta_j}k(y_{m},\theta_{j,m-1}) }{\sum_{i=1}^r w_{i,m-1}k(y_{m},\theta_{i,m-1})},\label{equ:update_rule_locations}
\end{align}
where $\alpha_{j,m}=\eta_m\, w_{j,m-1}, \beta_{j,m}=\eta_m\,g(w_{m-1},\theta_{m-1})$ is the step size, $\eta_m=1/(m+n+1)$ and $g$ is a to be assigned function. In summary, $\widehat{G}$ is the estimator and the uncertainty about it is derived by randomizing the weights and the kernel parameters by implementing the algorithm (\ref{equ:update_rule_weights}) and (\ref{equ:update_rule_locations}) from $m=n+1$ to a limit for which convergence for each $j$ is attained.  It is a natural extension of the BB which starts with the empirical distribution function and also randomizes the weights and the atoms using an algorithm which replaces the kernel with a single point at each atom.
A summary is provided in Algorithm~\ref{algo:The Bayesian Bootstrap for Mixture Models (BBM)}.

\section{Convergence and asymptotic exchangeability}
\label{sec:convergence_and_exchangeability}

In this section we establish properties of the BBM Algorithm~\ref{algo:The Bayesian Bootstrap for Mixture Models (BBM)} for constructing the UQ about the $\widehat{G}$. In particular, we are required to show the existence of $w_\infty$ and $\theta_{\infty}$, and that this implies the existence of a $G_\infty$ and a 
$f_\infty(y)=\int k(y\mid\theta)\,dG_\infty(\theta)$.

\subsection{Convergence}
\label{subsec:convergence}
The convergence of the (\ref{equ:update_rule_weights}) and (\ref{equ:update_rule_locations}) relies on it being a martingale, hence here we review the Doob's martingale theorem for a sequence $(z_m)$. 

\begin{theorem}[Doob's Martingale Convergence Theorem]
\label{thm_Doob}
If $(z_m)$ is a supermartingale, and
$\sup_m E\,|z_m|<\infty$, then $\lim_{m\to\infty}z_m=z_\infty$ exists almost surely and $z_\infty$ is finite in expectation.
\end{theorem}

The weights process converges since the martingale for each of them is clearly bounded. Hence, with probability 1, $w_\infty$ exists. However, for the new martingale, compared to the BB, we need to show convergence of the parameters.

First, in order to have zero-expectation score functions and hence a martingale, we need some regularity conditions; see, for example, \cite{Serfling_2009}. Let $\Theta$  be an interval in $\mathbb{R}^d$, and the update rule from (\ref{equ:update_rule_locations}). 
We assume:
\begin{enumerate}[label=(R.1),align= left]
\item \label{regularity:existance} 
For each $\theta \in \Theta$, the gradient $\nabla_{\theta} k(y,\theta)$ exists, for all $y$.
\end{enumerate}

\begin{enumerate}[label=(R.2),align= left]
\item \label{regularity:differentiable} 
There exists a function \(\phi(y)\), such that for each $s=1,\ldots,d$, $\left|\partial k(y , \theta)/\partial \theta_s\right| \leq \phi(y)$
holds for all $y$ and \(\theta \in \Theta\), and $\int \phi(y) d y<\infty$.
\end{enumerate}
If Conditions \ref{regularity:existance} and~\ref{regularity:differentiable} are satisfied, then by the Lebesgue dominant convergence theorem, we have 
$\int \nabla_{\theta}k(y,\theta)\mathrm{d}y=0,$
and so $(\theta_m)_{m>n}$ is a martingale.

The variance of $\theta_m$ is studied component-wise and so we focus on computing  $\mbox{Var}\,(\theta_{j,m}\mid y_{1:n})$. To this end we note that this is easily seen to be
$$\mbox{Var}\,(\theta_{j,m}\mid y_{1:n})=\sum_{i=n}^{m-1}
\beta_{j,i}^2\,w_{j,i}^2\,\int\frac{[k'(y\mid\theta_{j,i})]^2}{p_{i}(y)}\,dy.
$$
Now $p_i(y)$ is lower bounded by $w_{j,i}\,k(y\mid\theta_{j,i})$ and hence the integral is upper bounded by $I(\theta_{j,i})$, where $I(\theta)$ is the Fisher information for $k$ evaluated at $\theta$. Therefore, 
$$\mbox{Var}\,(\theta_{j,m}\mid y_{1:n})\leq \sum_{i=n}^{m-1}
\beta_{j,i}^2\,w_{j,i}\,I(\theta_{j,i}).$$
In order to determine the variance, we consider two separate cases.  The first is when $\Theta$ is compact and the second when $\Theta$ is not compact.
When $\Theta$ is compact, and we can assume this is arbitrarily large so the sequences $(\theta_{j,m})$ remain inside, we assume that $I(\theta)$ is continuous on $\Theta$. Hence, by the extreme value theorem, $I(\theta)$ is bounded on $\Theta$.  Taking $g\equiv 1$ we see that the variance as $m\to\infty$ is bounded above by
$\sum_{i=1}^\infty \max_{\Theta} I(\theta)/(i+n)^2<\infty$.

When $\Theta$ is no longer compact, we can no longer assume that the Fisher information is bounded. This is where we adapt the step-size to include $g$. Specifically, we take
$$g(w_{j,m},\theta_{j,m})=\frac{1}{\sqrt{I(\theta_{j,m})\,I(\theta_{j,n})\,w_{j,m}}}.$$
Now we can see that the variance in the limit as $m\to\infty$ is upper bounded by
$$\sum_{i=1}^\infty \frac{1}{(i+n)^2}\frac{1}{I(\theta_{j,n})},$$
which is approximately $\{n\,I(\theta_{j,n})\}^{-1}$, which is an objective type of variance for a posterior distribution. 

Under further assumptions, we have weak convergence of $P_m(y)=\int K(y,\theta) d\,G_m(\theta)$ almost surely, where $K$ is the distribution function corresponding to density $k$. For the following, set 
$$\theta_m=(\theta_{1,m},\ldots,\theta_{r,m})\in\mathbb{R}^{d\times r} \text{\quad and\quad} w_{m}=(w_{1,m},\ldots,w_{r,m})\in\mathbb{R}^{1\times r}.$$
We use the following result from \cite{Berti2006}: If $\mu_m$ is a sequence of random probability measures and $\mu_m(A)\to \mu_\infty(A)$ almost surely for all relevant sets $A$, and the sequence of probability measures $E\,\mu_m$ is tight, then $\mu_m\to \mu_\infty$ weakly almost surely. See also \cite{Lijoi2007}.

We also use Doob's martingale inequality, namely that if $(\psi_m)$ is a non-negative sub-martingale; i.e. $E\,(\psi_{m+1}\mid y_{1:m})\geq \psi_m$, then
$$\P\left(\max_{1\leq i\leq m} \psi_i>C\right)\leq C^{-1}\psi_0,$$
where $\psi_0=E\,\psi_m$.

\begin{lemma}
\label{lemma:G_tight}
    Assume $\{\theta_m\}$ and $\{w_m\}$ are convergent martingales with bounded variances. Then the sequence of probability measures $\{\bar{G}_m\equiv E(G_m\mid y_{1:n})\}$ is tight on $(\mathbb{R}^d,\mathcal{B}(\mathbb{R}^d))$.
\end{lemma}
\begin{proof}
    Recall, $G_m(\Theta)=\sum_{j=1}^r w_{j,m}1(\theta_{j,m}\in\Theta)$ for any set $\Theta\subset\mathbb{R}^d$  and so
    $$\bar{G}_m(\Theta)\leq \sum_{j=1}^r P(\theta_{j,m}\in\Theta). $$
    We take $\Theta_\epsilon=(-\infty, c_\epsilon)\cup (c_\epsilon,\infty)$ and, noting that $(\theta^2_{j,m})$ is a sub-martingale, from Doob's martingale inequality it is that
    $\bar{G}_m(\Theta)\leq r\,\max_j\{\bar\theta_{j,n}^2\}/c_\epsilon^2,$
    where $\bar{\theta}_{j,n}^2=E\,\theta_{j,m}^2$. We have previously established that these are finite for all $j$ and all $m$, so $\bar{G}_m(\Theta_\epsilon)\leq C/c_\epsilon^2$ for some $C<\infty$. We can therefore choose $c_\epsilon$ so $\bar{G}_m(\Theta_\epsilon)<\epsilon$ for all $m$, indicating the sequence is tight. 
\end{proof}

\begin{lemma}
\label{lemma:p_tight}
    Assume $G_m$ converges weakly to $G_\infty$ almost surely and the sequence $(\bar{G}_m)$ is tight. Further, assume for any $\delta>0$ there exists a set $A_\delta$ and a $c_\delta$ such that $K(A_\delta\mid \theta)>1-\delta$ for all $|\theta|<c_\delta$. Then
    $(P_m)$ converges weakly to $P_\infty$ with probability one.
\end{lemma}

\begin{proof}
    We use similar ideas with the proof of Lemma~\ref{lemma:G_tight}. Set $P_m$ and $K$ as the corresponding probability distributions for $p_m$ and $k$, respectively. Since $K(y\mid \theta)$ is bounded and continuous in $\theta$ for all $y$, it follows that
    $P_m(y)\to P_\infty(y)$ almost surely for each $y$. To show tightness we consider the sequence $\bar{P}_m=\int K(y\mid\theta)\,d\bar{G}_m(\theta)$.
    
     For any $\epsilon>0$, there exists a large enough $a_\epsilon$ such that $\bar{G}_m[-a_\epsilon,a_\epsilon]>\sqrt{1-\epsilon}$, see Lemma~\ref{lemma:G_tight}, and also take $a_\epsilon$ large enough so that there exists a compact set $A_\epsilon$ for which $K(A_\epsilon\mid\theta)>\sqrt{1-\epsilon}$  for all $|\theta|<a_\epsilon$. This is possible based on the assumption stated in the Lemma. Hence, $\bar{P}_m(A_\epsilon)>1-\epsilon$ for all $m$.
    So $\{\bar{P}_m\}$ is tight and therefore $(P_m)$ converges weakly to $P_\infty$ almost surely.
\end{proof}

The usual BB provides an exchangeable sequence for the $x_{n+1:\infty}$. According to the BBM Algorithm~\ref{algo:The Bayesian Bootstrap for Mixture Models (BBM)}, if only the weights get updated and the $\theta$ remain fixed, then the sequence $(y_{n+1:\infty}$
form a conditionally identically distributed (c.i.d.) sequence and are therefore asymptotically exchangeable; see \cite{Aldous_1985}, \cite{Berti_2004} and \cite{Fortini_2020}. This follows for the $y_{n+1:\infty}$ since the $(P_m)_{m>n}$ form a martingale sequence when the $\theta$ stay fixed. This guarantees the c.i.d. sequence since $E\,(P_m\mid P_l)=P_l$ for all $m>l\geq n$.

It is interesting to note that when only the weights change, the updating algorithm is identical to the Newton algorithm when started off at a discrete distribution; see \cite{Newton_1998} and ~\cite{Newton_1999}. These articles used the algorithm to estimate a distribution from the data rather than to use it for uncertainty quantification as we do. Moreover, we extend the algorithm to update the $\theta$. 
The Newton algorithm is given by, for $i=1,\ldots,n-1$, 
\begin{equation}\label{newton}
G_{i+1}(\theta)=(1-\eta_i)\,G_i(\theta)+\eta_i
\frac{\int_{-\infty}^\theta k(y_{i+1}\mid s)\,dG_i(s)}
{\int k(y_{i+1}\mid s)\,dG_i(s) },
\end{equation}
where the $(y_i)_{i=1:n}$ represent the observed data. Hence, if the start $G_0$ is discrete, it is seen that only the weights will get updated, the locations will not.

When the $\theta$ get updated alongside the weights, the sequence is no longer c.i.d. However, given that $(P_m)$ converges weakly almost surely, the asymptotic exchangeability remains.

\begin{definition}[Asymptotic Exchangeability]
\label{def:asymptotic_exchangeability}
    A sequence of random variables \(\left(y_{m}\right)\) is asymptotically exchangeable, if
$$
\left(y_{m+1}, y_{m+2}, \ldots\right) \stackrel{d}{\rightarrow}\,z=\left(z_{1}, z_{2}, \ldots\right)
$$
    for some exchangeable sequence \(z\) as $m\to\infty$.
\end{definition}

\begin{theorem}
\label{thm:Asymptotic Exchangeability}
    Suppose $y_{m+1}\sim p_{m}$ for every $m=n+1,n+2,\ldots$, where $p_{m}$ is updated using (\ref{equ:update_rule_weights})-(\ref{equ:update_rule_locations}). If $(P_m)$ converges weakly almost surely then the sequence $(y_{n+1},y_{n+2},\ldots)$ is asymptotically exchangeable.
\end{theorem}

The proof of the theorem is obtained by the almost sure convergence of $P_m$ according to  Lemma 8.2(b) in~\cite{Aldous_1985}.

Finally, in this section, we remark that if the distribution estimator from the data is taken to be continuous, using for example the Newton algorithm or some other approach (e.g. as in \cite{Chae_2018}), and the updating for $(\widehat{G}_m)_{m>n}$ is done using the Newton algorithm, i.e. (\ref{newton}), then the sequence $(\widehat{G}_m)$ is easily seen to be a martingale and hence the sequence $(y_{n+1:\infty})$ is c.i.d. An illustration of this is given in the next section.

\section{Illustrations}
\label{sec:illustration}

We start this section with simulated data using a normal kernel with a known variance. When the variance is unknown, the NPMLE fails since the kernel is unbounded as the variance goes to 0. In this case we need an alternative to the NPMLE. For this we use BIC to provide the number of components and then maximum likelihood for the locations and common variance estimators. The algorithm then proceeds as usual with this estimator. This is demonstrated on the well known Galaxy dataset which is typically modeled with an unknown common variance. Finally, we use another real dataset where the assumption is that the density is monotone decreasing on $(0,\infty)$ and so we use an exponential kernel. This also is unbounded and in this case we use the Newton algorithm, see (\ref{newton}) to provide the initial estimator of the mixing distribution. In this case the mixing distribution will be continuous rather than discrete.

\subsection{Simulated data}

For discrete $G$, we generate $n=100$ and $n=500$ samples independently from the Gaussian mixture model: $f(y)=\sum_{j=1}^3 w_j\,\N(y\mid \theta_j,0.1^2)$, where $\theta = (1, 3, 5)$, and $w = (0.2, 0.5, 0.3)$. We program in R with the package ``REBayes"~\cite{REBayes_R} to compute the Kiefer-Wolfowitz NPMLE for Gaussian location mixtures, \cite{Jiang_2009}, which is the same as~\cite{Chae_2018} when the kernel is Gaussian, and keep every atom whose weight is larger than $10^{-4}$, combined then with a normalization of the remaining weights. For continuous $G$, we generate $n=50$ and $n=500$ samples independently from the model $f(y;G)=\int k(y;\theta)\,dG(\theta),$ where $k(y;\theta)=\N(\theta, 0.1^2)$, and $G(\theta)$ is either standard normal, i.e. $\N(0, 1)$ or gamma, i.e. $\mathrm{Ga}(5,2)$.

We then complete BBM algorithms with 10000 iterations for each bootstrap sample and draw the CDFs of the results with different sample sizes, together with the CDFs of the original NPMLEs and the true distribution. See Figure \ref{pl:CDF_dis_loc_gradients_REBayes}. In Figure \ref{pl:CDF_dis_100_loc_gradient_REBayes}, we show the CDFs of the true sampling discrete distribution. The NPMLE has the number of atoms $r=7$ with sample size $n=100$. In Figure \ref{pl:CDF_dis_500_loc_gradient_REBayes}, we show the similar information with sample size $n=500$ and in this case $r=11$. Figure~\ref{pl:CDF_normal_50_loc_gradient_REBayes} and~\ref{pl:CDF_normal_500_loc_gradient_REBayes} shows 100 bootstrap samples of the CDFs for $G=\N(0,1)$ with $r=21$ and $r=38$, respectively, for sample sizes $n=50$ and $n=500$. Figure~\ref{pl:CDF_gamma_50_loc_gradient_REBayes} and~\ref{pl:CDF_gamma_500_loc_gradient_REBayes} shows 100 bootstrap samples of the CDFs for $G=\mathrm{Ga}(5,2)$ with $r=19$ and $r=35$, respectively, for sample sizes $n=50$ and $n=500$.

\begin{figure}[ht!]
\centering
\subfigure[]{
\label{pl:CDF_dis_100_loc_gradient_REBayes}
\includegraphics[width=0.45\textwidth]{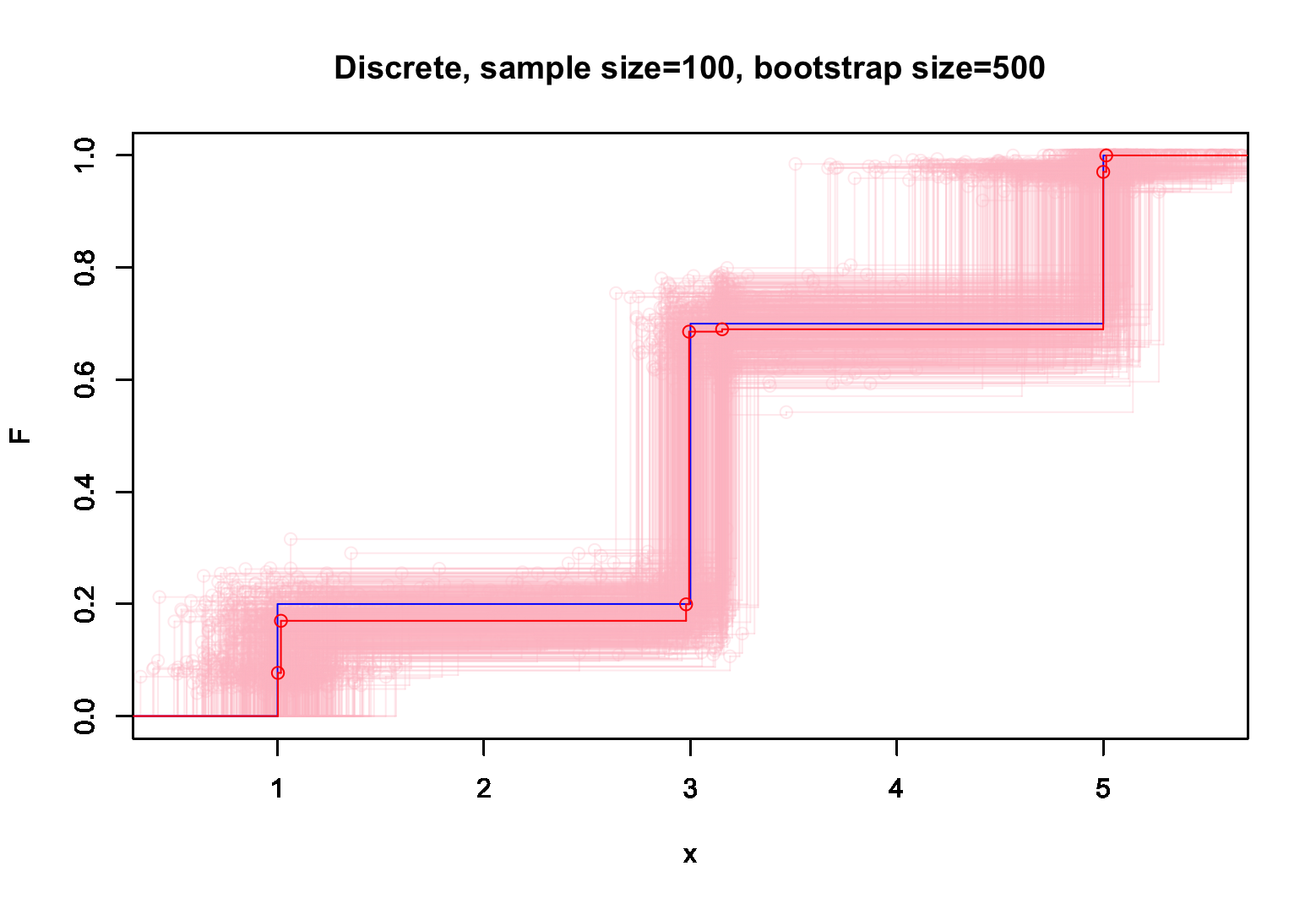}}
\subfigure[]{
\label{pl:CDF_dis_500_loc_gradient_REBayes}
\includegraphics[width=0.45\textwidth]{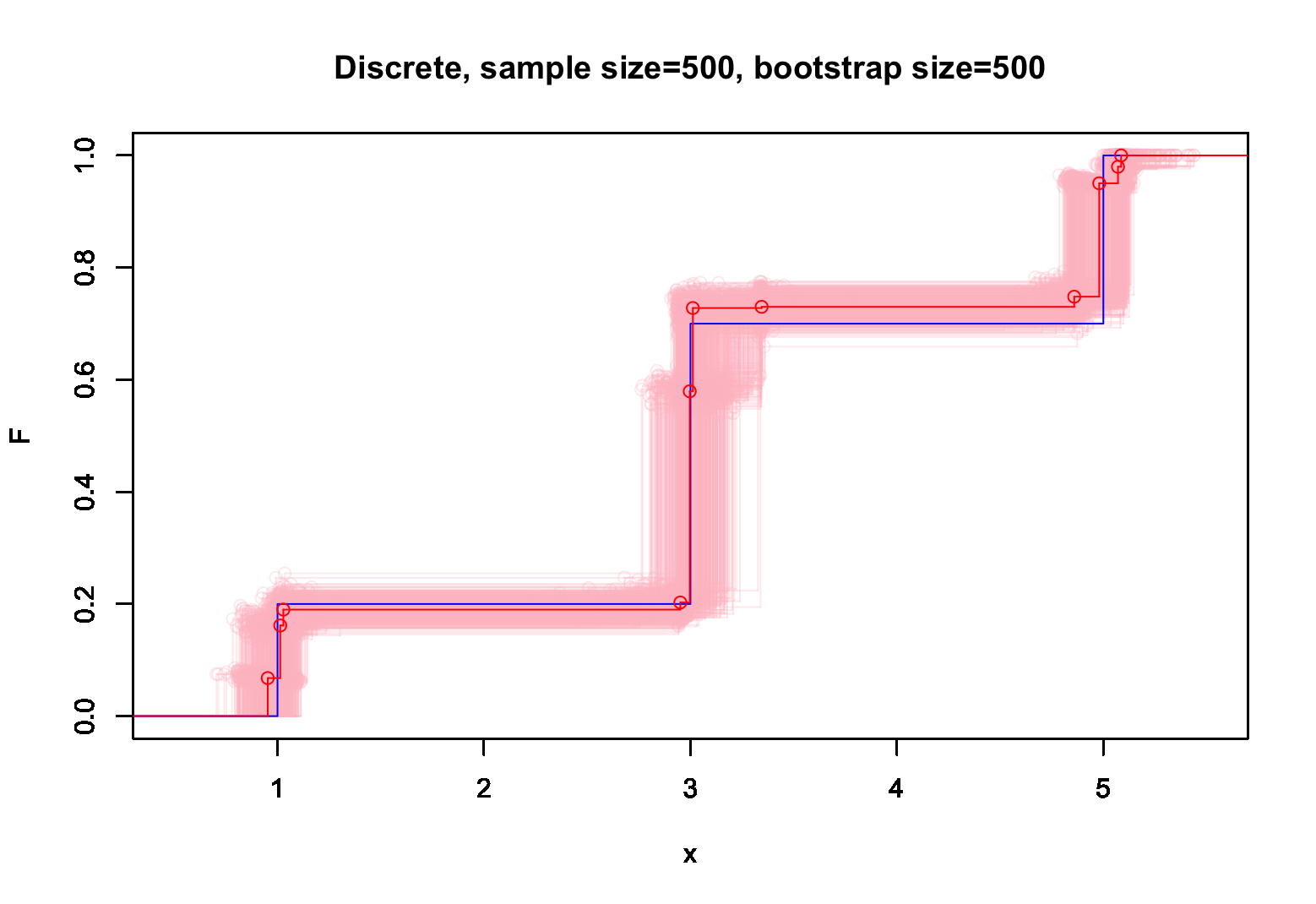}}
\subfigure[]{
\label{pl:CDF_normal_50_loc_gradient_REBayes}
\includegraphics[width=0.45\textwidth]{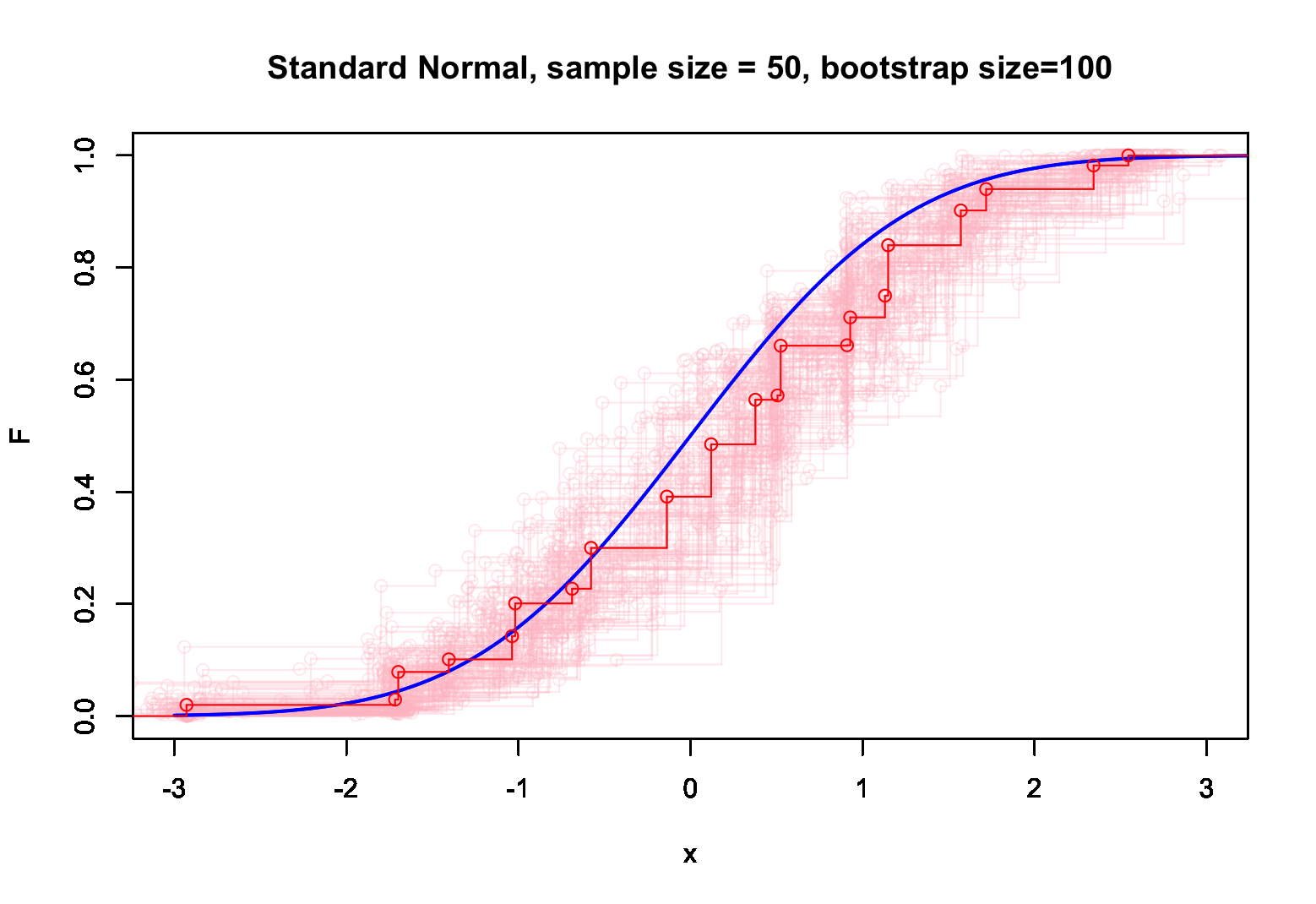}}
\subfigure[]{
\label{pl:CDF_normal_500_loc_gradient_REBayes}
\includegraphics[width=0.45\textwidth]{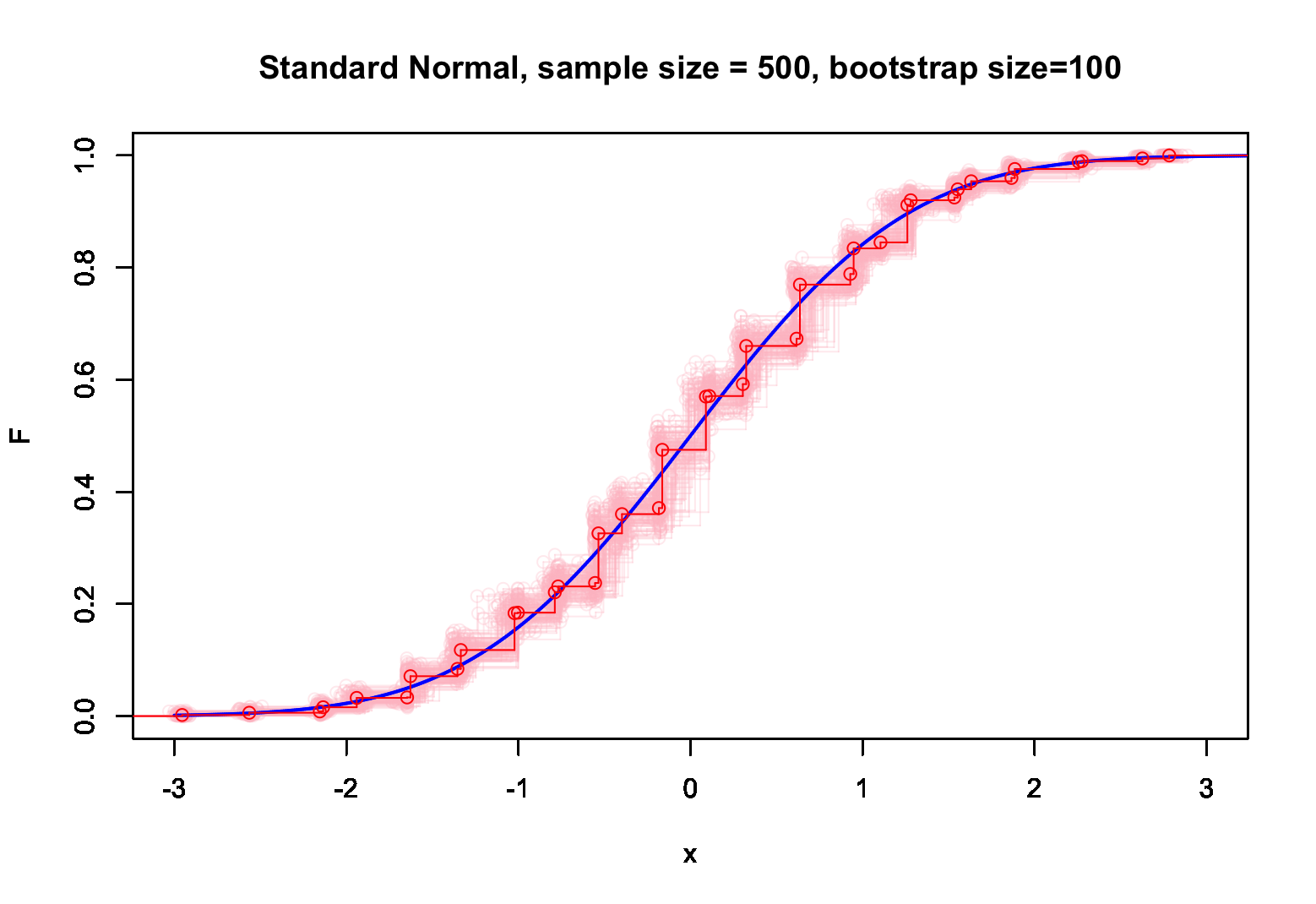}}
\subfigure[]{
\label{pl:CDF_gamma_50_loc_gradient_REBayes}
\includegraphics[width=0.45\textwidth]{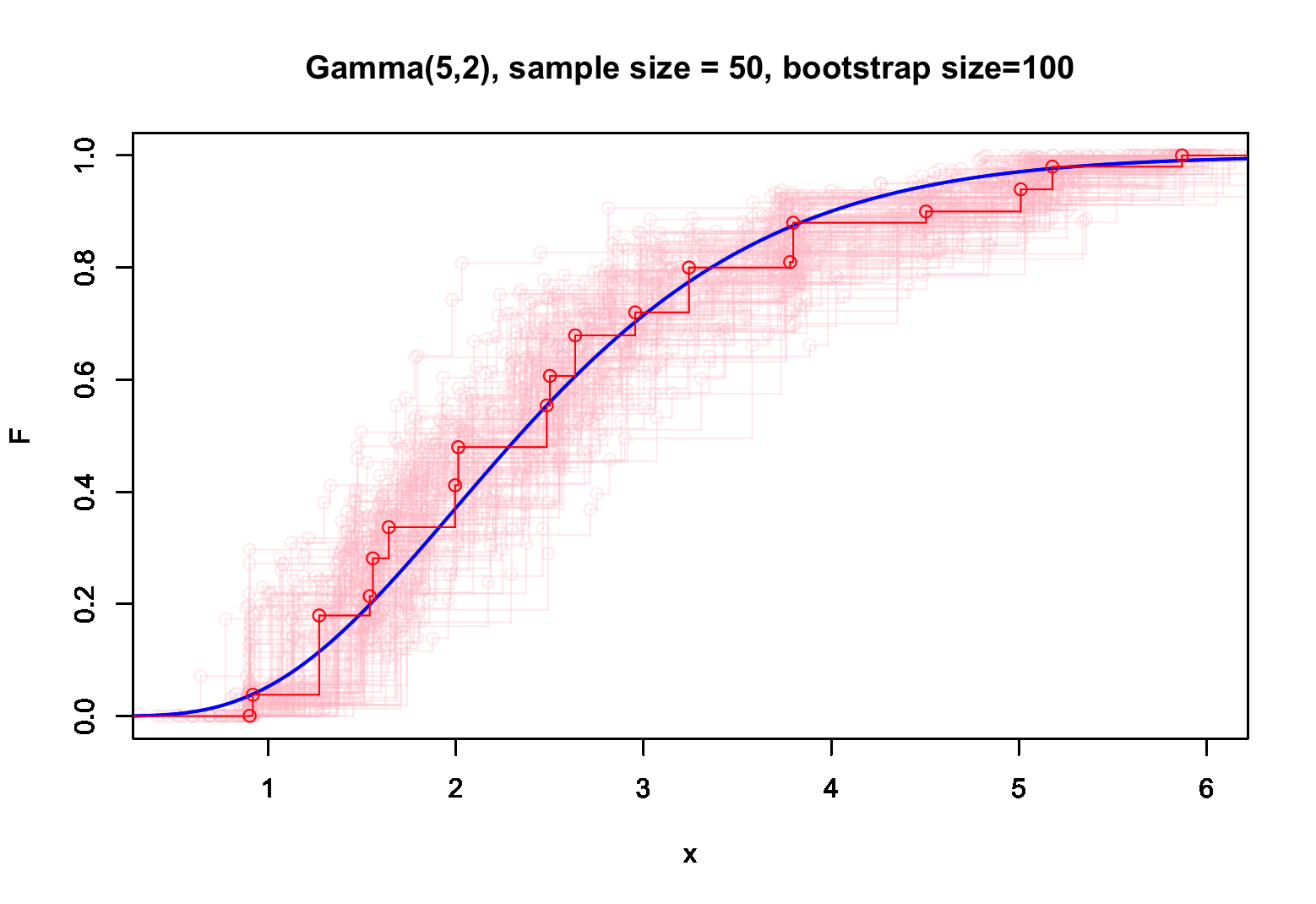}}
\subfigure[]{
\label{pl:CDF_gamma_500_loc_gradient_REBayes}
\includegraphics[width=0.45\textwidth]{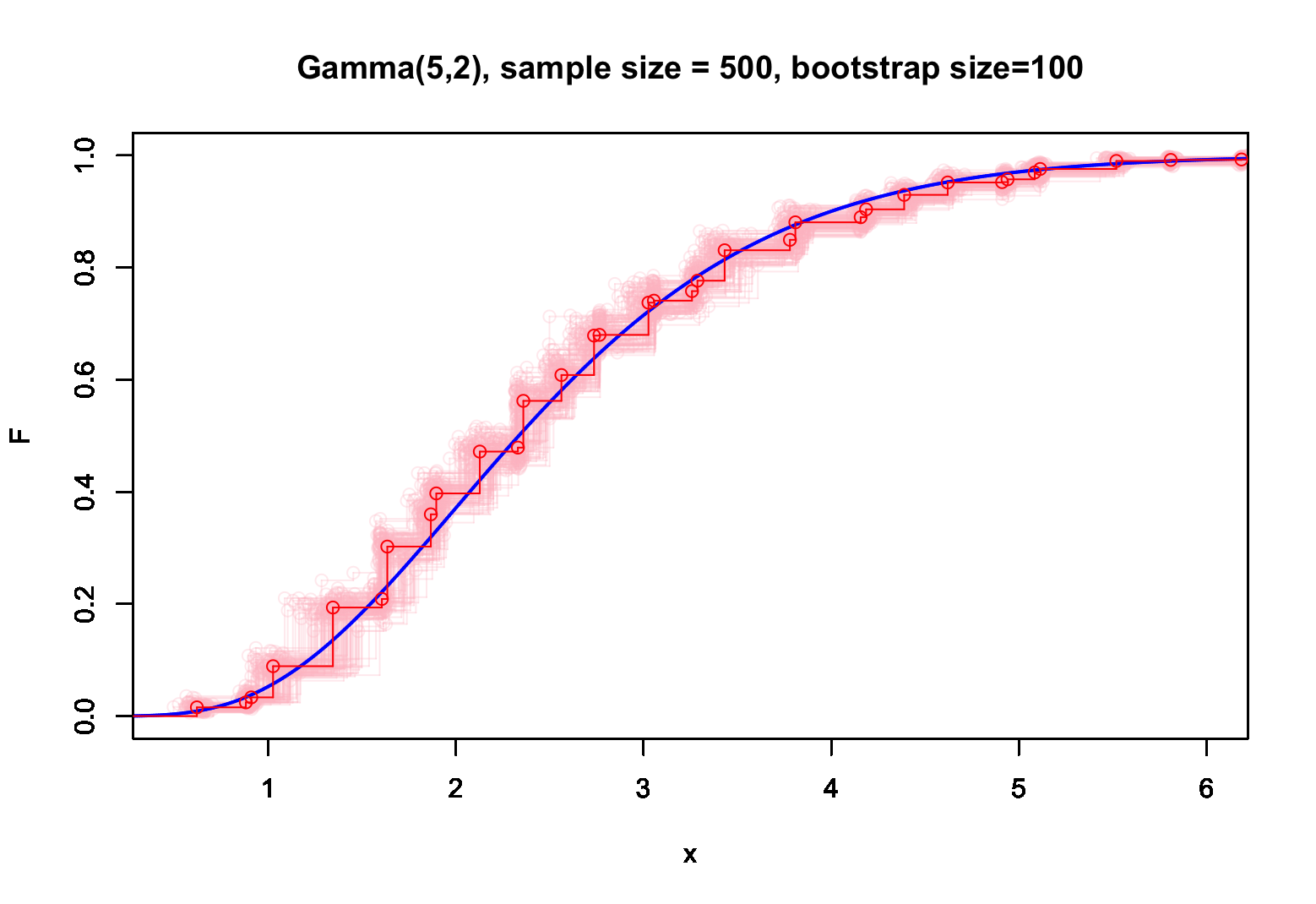}}
\caption{CDFs of BBM results using our gradient schemes. In both plots, the blue lines are CDFs of the true distribution, i.e. $G$, the red lines are CDFs of the NPMLEs, and the pink lines are CDFs of the martingale bootstrapping samples. The red and pink points represent the supports of the NPMLEs' and their bootstrapping samples' probability mass functions.}
\label{pl:CDF_dis_loc_gradients_REBayes}
\end{figure}

We see that the true distribution is included in the ranges of the bootstrap samples for both sample sizes and for all models, even when the NPMLE is not so close to the true distribution. As expected, the range of random CDFs is less for the larger sample sizes, indicating a less uncertainty.

\subsection{Galaxy data set}

The Gaussian kernel is also problematic for NPMLE when mixing over both the mean and variance, since the family of normal density functions is unbounded when the variance goes to 0. This is not solved by arbitrarily imposing a lower bound for the variance, since the NPMLE will simply pick out this variance to use as a point estimator. However, there is no compelling obligation to start off with the NPMLE and so, in this example, rather than using a Gaussian kernel mixed over both mean and variance, we assume a common variance for each normal component, a standard procedure, and we  use an information criterion to select the number of components. See, for example, \cite{Chen_1998} and \cite{McLachlan_2014}. Here under our Bayesian schemes, we naturally choose Bayesian information criterion (BIC).

Here we use a popular real data set of size $n=82$, which includes the velocities of 82 galaxies, see \cite{Venables_2002} and see Figure~\ref{fig:GMM_hist} for the histogram of the data set.
We assume a Gaussian kernel $k(y|\mu,\sigma^2)=\N(\mu,\sigma^2)$ with both $\mu$ and $\sigma^2$ unknown. If using NPMLE to get the estimate, $\sigma^2$ will go to zero and the number of kernels will become $n=82$.

\begin{center}
\begin{figure}[!htbp]
\begin{center}
\includegraphics[width=12cm,height=5cm]{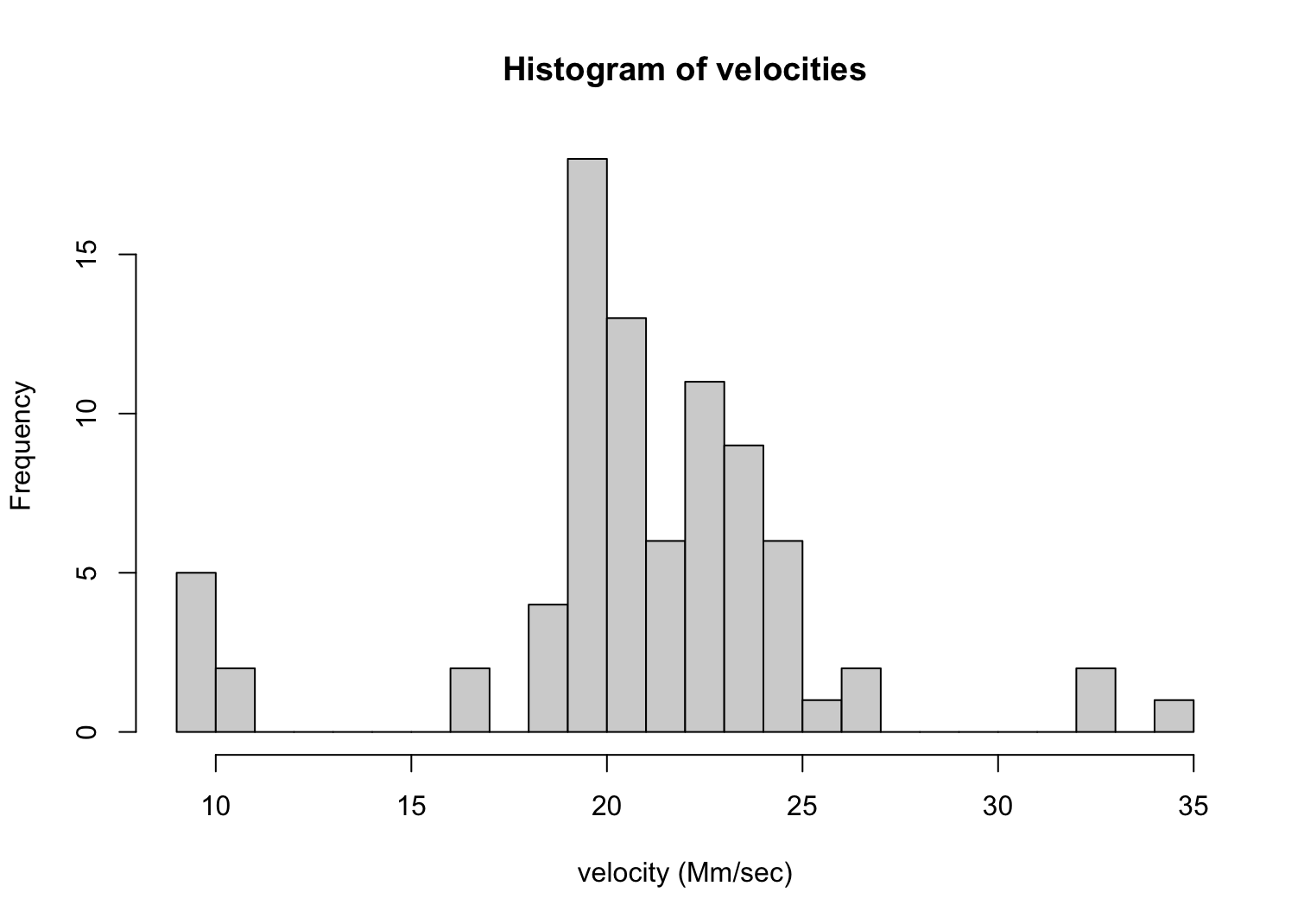}
\caption{Histogram of velocities of galaxies}
\label{fig:GMM_hist}
\end{center}
\end{figure}
\end{center}  

The mixture model is given by
$$
p_{\sigma,G}(y)=\int k(y|\mu,\sigma^2)\, dG(\mu),
$$
implying the variances $\sigma^2$ is the same for all kernels. We use information criteria to choose the number of kernels, which will provide the number of atoms for $G$. With this, we can then estimate the corresponding weights and atoms of $G$ and the variance $\sigma^2$ using standard algorithms, such as EM.

We use BIC, with criterion given by
\begin{align*}
    \mathrm{BIC} &= 2\log(\mathrm{Likelihood})-\log(n)d,
\end{align*}
where $d$ is the degrees of freedom in the model. We will find the maximum value of the information criterion to determine the number of atoms.
In the example (see Figure~\ref{fig:BIC}), BIC shows that the optimal number is 3. 

\begin{center}
\begin{figure}[!htbp]
\begin{center}
\includegraphics[width=12cm,height=5cm]{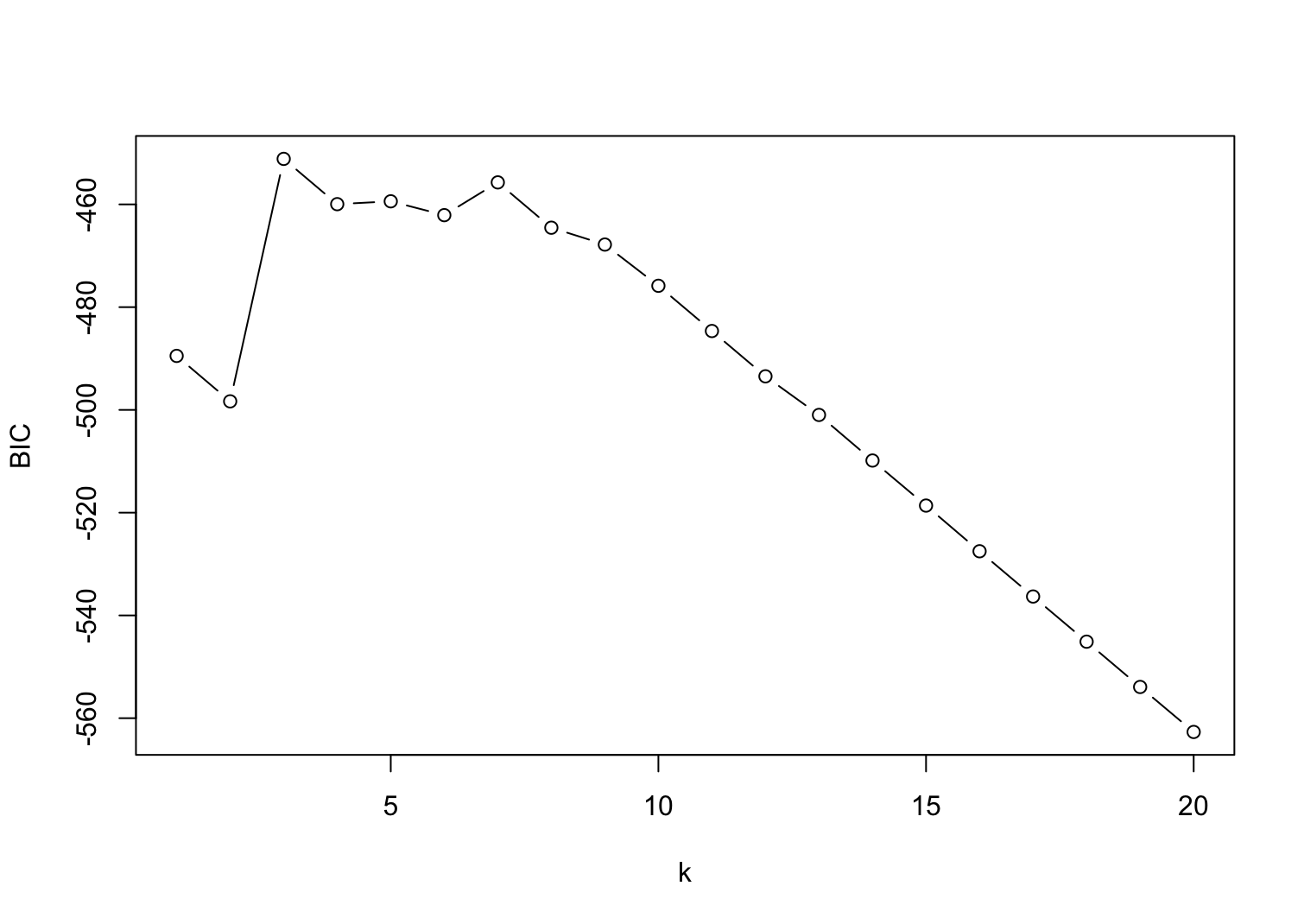}
\caption{BIC for Gaussian mixture model and Galaxy data}
\label{fig:BIC}
\end{center}
\end{figure}
\end{center} 

We can now use one of two schemes to get the weights and atoms for $G$ and $\sigma$. The first one is to use the EM algorithm. The alternative is estimate $\hat{\sigma}^2$ from the EM algorithm, and then to use the NPMLE for $G$ using the variance estimator. For ease and to ensure we recover the BIC number of components, we use the former of the two. 

The BBM procedure can again be regarded as the extension of Algorithm~\ref{algo:The Bayesian Bootstrap for Mixture Models (BBM)}, to incorporate the updating of $\sigma$. So $y_m$ is sampled from $p_{m-1}(y)=\int k(y\mid \mu,\sigma_{m-1}^2)\,d G_{m-1}(\mu)$ and the locations and weights of $G_{m-1}$ get updated as previously described. For the common variance, use the score function of $\sigma^2$ to perform the update, i.e.
$$
\sigma^2_{m} = \sigma^2_{m-1}+\beta_{m-1} \frac{\sum_{j=1}^r w_{j,m-1}\frac{\partial}{\partial \sigma} k(y_m\mid\mu_{j,m-1},\sigma^2_{m-1})}{\sum_{j=1}^r w_{j,m-1} k(y_m\mid\mu_{j,m-1},\sigma^2_{m-1})}.
$$
Figure~\ref{fig:GMM} presents 500 samples of random density functions, i.e. $p_\infty$, from the BBM method.

\begin{center}
\begin{figure}[!htbp]
\begin{center}
\includegraphics[width=12cm,height=5cm]{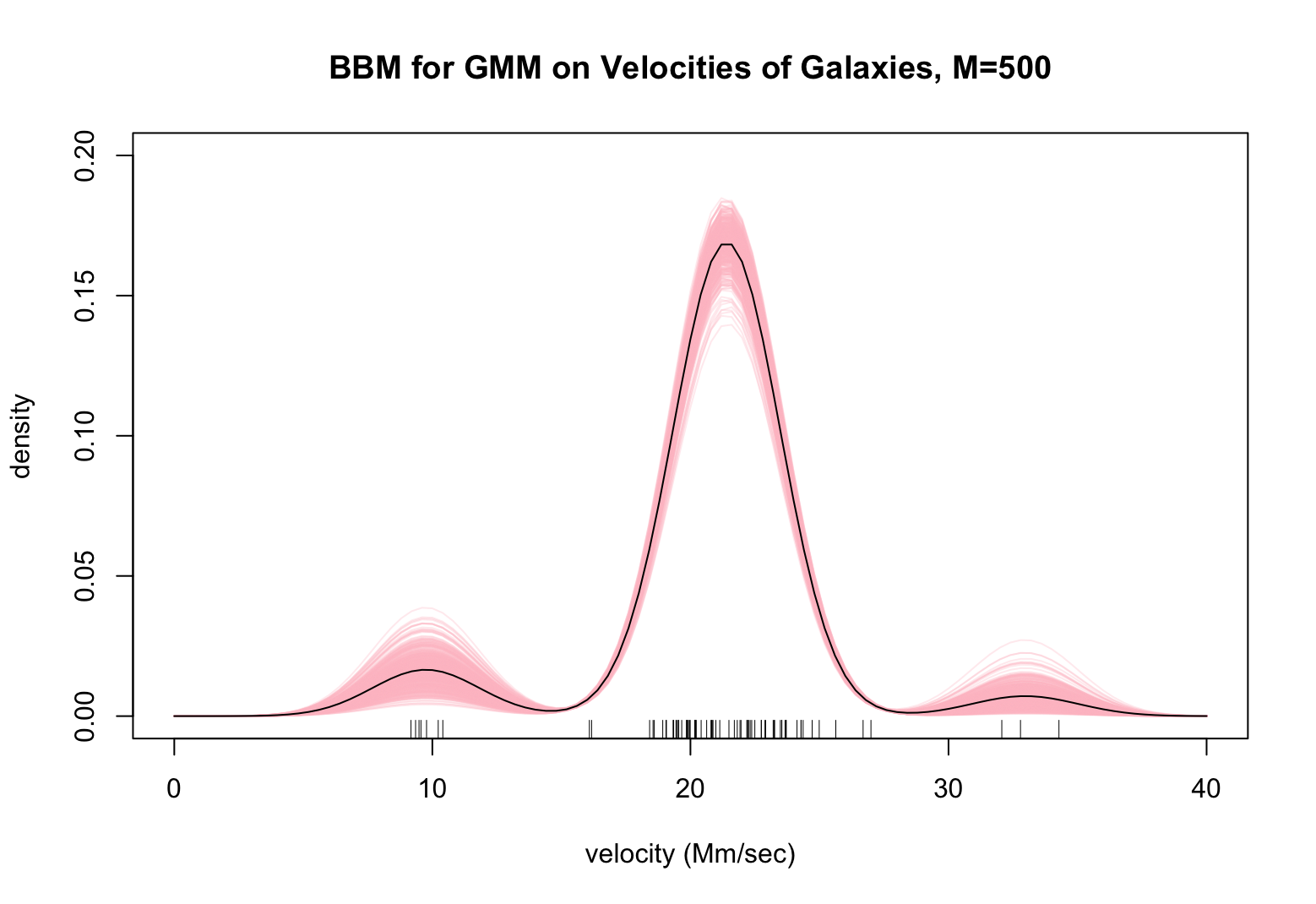}
\caption{Random distribution functions $p$ from BBM on velocities of galaxies}
\label{fig:GMM}
\end{center}
\end{figure}
\end{center}

\begin{center}
\begin{figure}[!htbp]
\begin{center}
\includegraphics[width=12cm,height=5cm]{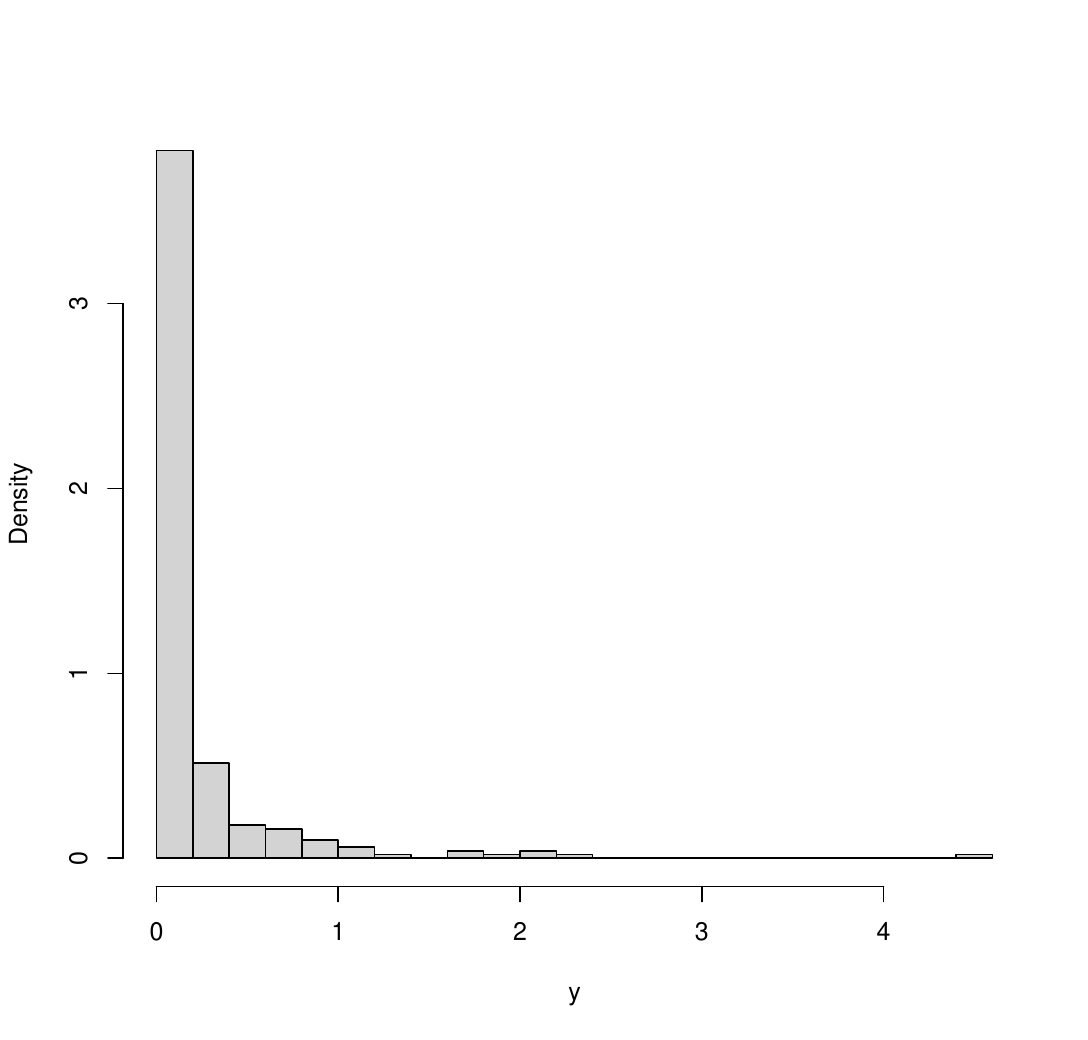}
\caption{Histogram of square log returns}
\label{fign}
\end{center}
\end{figure}
\end{center}

\subsection{Real data: squared log-returns} In this section we analyse a real data set of size $n=252$. A histogram of the data is given in Fig.~\ref{fign} which involves a sample of size $n=252$ of square log returns. We assume the density is decreasing and therefore use a mixture of exponential model. Since the exponential model; i.e. $k(y\mid\theta)=\exp(-y/\theta)/\theta$ is unbounded, we are unable to obtain the NPMLE as the initial estimator. Instead we employ the Newton algorithm given in (\ref{newton}) in order to obtain the estimator. The bootstrap procedure now would be a continuation of the Newton algorithm for $m>n$ to obtain the martingale sequence $(G_m)$, whereby $y_m$ is sampled from
$p_{m-1}(y)=\int k(y\mid\theta)\,dG_{m-1}(\theta)$ and then we update $G_m$ to $G_{m+1}$ using the $y_{m}$ sampled and the update provided by (\ref{newton}).

\begin{center}
\begin{figure}[!htbp]
\begin{center}
\includegraphics[width=12cm,height=5cm]{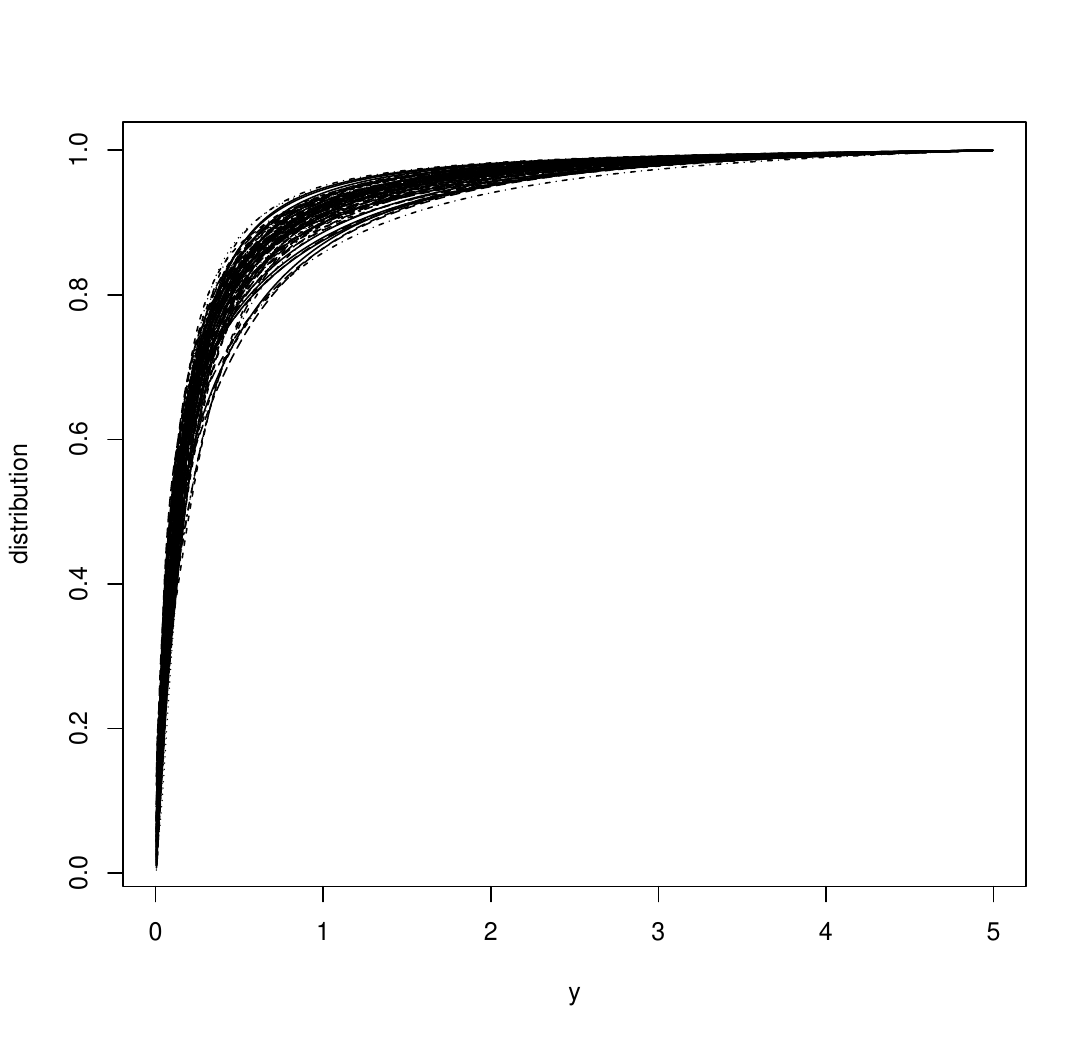}
\caption{Random distribution functions $G$ from BBM}
\label{figo}
\end{center}
\end{figure}
\end{center}  

In Fig.~\ref{figo} we present 100 samples of random distribution functions for $G$ from the BBM. Each run actually does not start at the same Newton estimator, but randomizes the order of the data to get a slightly different estimator for each random distribution. This is to remove the dependence on the order arbitrarily assigned to the data.

\section{Summary and Discussion}
\label{sec:summary_and_discussion}
In this paper we have expanded the class of Bayesian bootstrap methods to cover a nonparametric mixture model. The idea is developed from the P\'olya-urn model and the original Bayesian bootstrap, by representing them as a stochastic gradient optimization solver, along with the construction of martingales. The difference in practice between the BB and BBM is the switch from a point mass kernel to a continuous kernel. There are some other necessary adaptions but the essence of the algorithm remains the same. 
In short, the key is to see the P\'olya-urn process as a stochastic gradient optimization solver. 

In slightly more detail, the BB updates weights using a stochastic gradient algorithm and a point mass kernel, while fixing the number of atoms. The number of atoms when relevant also remain fixed. When we introduce a kernel, the weights are updated as with the BB, though now we also need to update the location of the kernels, which is also done using a stochastic gradient algorithm.

If the kernel is bounded we can use the NPMLE estimator for the data estimated distribution $\widehat{G}$. When it is unbounded we need an alternative estimator for which we use either the Newton algorithm or the BIC. If we use the former, which will result in a continuous $\widehat{G}$, the Newton algorithm for generating and updating from $m>n$ can be used, which results in $(y_{n+1:\infty})$ being a conditionally identically distributed sequence.

It is quite evident that the uncertainty quantification approach using the BBM is far more straightforward than the use of say MCMC methods. Moreover, the BBM output can be coded in parallel. From a theoretical perspective,  we have shown the convergence and (asymptotic) exchangeability of the algorithms.

\bibliographystyle{abbrv}
\bibliography{main_version}

\end{document}

%% file: final_macros.tex
\newtheoremstyle{named}{}{}{\itshape}{}{\bfseries}{.}{.5em}{\thmnote{#3's }#1}
\theoremstyle{named}

\theoremstyle{plain}

\newtheorem{theorem}{Theorem}

\newtheorem{lemma}{Lemma}

\newtheorem{definition}{Definition}

\newlength{\widebarargwidth}
\newlength{\widebarargheight}
\newlength{\widebarargdepth}

\makeatletter
\long\def\@makecaption#1#2{
        \vskip 0.8ex
        \setbox\@tempboxa\hbox{\small {\bf #1:} #2}
        \parindent 1.5em  
        \dimen0=\hsize
        \advance\dimen0 by -3em
        \ifdim \wd\@tempboxa >\dimen0
                \hbox to \hsize{
                        \parindent 0em
                        \hfil
                        \parbox{\dimen0}{\def\baselinestretch{0.96}\small
                                {\bf #1.} #2
                                }
                        \hfil}
        \else \hbox to \hsize{\hfil \box\@tempboxa \hfil}
        \fi
        }
\makeatother

\long\def\comment#1{}
\definecolor{battleshipgrey}{rgb}{0.52, 0.52, 0.51}
\definecolor{darkgray}{rgb}{0.66, 0.66, 0.66}
\definecolor{darkgreen}{rgb}{0.0, 0.2, 0.13}
\definecolor{darkspringgreen}{rgb}{0.09, 0.45, 0.27}
\definecolor{dukeblue}{rgb}{0.0, 0.0, 0.61}
\definecolor{olivedrab7}{rgb}{0.24, 0.2, 0.12}
\definecolor{darkblue}{rgb}{0.0, 0.0, 0.55}
\definecolor{darkscarlet}{rgb}{0.34, 0.01, 0.1}
\definecolor{candyapplered}{rgb}{1.0, 0.03, 0.0}
\definecolor{ao(english)}{rgb}{0.0, 0.5, 0.0}
\definecolor{applegreen}{rgb}{0.55, 0.71, 0.0}


%% file: main_version.bbl
\begin{thebibliography}{10}

\bibitem{Aldous_1985}
D.~J. Aldous.
\newblock Exchangeability and {R}elated {T}opics.
\newblock In P.~L. Hennequin, editor, {\em {\'E}cole d'{\'E}t{\'e} de Probabilit{\'e}s de Saint-Flour XIII --- 1983}, pages 1--198, Berlin, Heidelberg, 1985. Springer Berlin Heidelberg.

\bibitem{Barrientos_2020}
A.~F. Barrientos and V.~Pe{\~n}a.
\newblock {Bayesian Bootstraps for Massive Data}.
\newblock {\em Bayesian Analysis}, 15(2):363 -- 388, 2020.

\bibitem{Berti_2004}
P.~Berti, L.~Pratelli, and P.~Rigo.
\newblock {Limit theorems for a class of identically distributed random variables}.
\newblock {\em The Annals of Probability}, 32(3):2029 -- 2052, 2004.

\bibitem{Berti2006}
P.~Berti, L.~Pratelli, and P.~Rigo.
\newblock Almost sure weak convergence of random probbaility measures.
\newblock {\em Stochastics}, 78:91--97, 2006.

\bibitem{Blackwell_1973}
D.~Blackwell and J.~B. MacQueen.
\newblock {Ferguson distributions via P\'olya-urn schemes}.
\newblock {\em The Annals of Statistics}, 1(2):353 -- 355, 1973.

\bibitem{Chae_2018}
M.~Chae, R.~Martin, and S.~G. Walker.
\newblock Convergence of an iterative algorithm to the nonparametric mle of a mixing distribution.
\newblock {\em Statistics \& Probability Letters}, 140:142--146, 2018.

\bibitem{Chen_1998}
S.~S. Chen and P.~Gopalakrishnan.
\newblock Clustering via the {B}ayesian information criterion with applications in speech recognition.
\newblock In {\em Proceedings of the 1998 IEEE International Conference on Acoustics, Speech and Signal Processing, ICASSP '98 (Cat. No.98CH36181)}, volume~2, pages 645--648 vol.2, 1998.

\bibitem{Favaro_2013}
S.~Favaro and Y.~W. Teh.
\newblock {MCMC for normalized random measure mixture models}.
\newblock {\em Statistical Science}, 28(3):335 -- 359, 2013.

\bibitem{Ferguson_1973}
T.~S. Ferguson.
\newblock {A Bayesian analysis of some nonparametric problems}.
\newblock {\em The Annals of Statistics}, 1(2):209 -- 230, 1973.

\bibitem{Fong_2021}
E.~Fong, C.~Holmes, and S.~G. Walker.
\newblock Martingale posterior distributions (with discussion).
\newblock {\em Journal of the Royal Statistical Society, Series B}, 2023.

\bibitem{Fortini_2020}
S.~Fortini and S.~Petrone.
\newblock {Quasi-Bayes properties of a procedure for sequential learning in mixture models}.
\newblock {\em Journal of the Royal Statistical Society Series B: Statistical Methodology}, 82(4):1087--1114, 06 2020.

\bibitem{Gershman_2012}
S.~J. Gershman and D.~M. Blei.
\newblock A tutorial on {B}ayesian nonparametric models.
\newblock {\em Journal of Mathematical Psychology}, 56(1):1--12, 2012.

\bibitem{Ghosal_2010}
S.~Ghosal.
\newblock {\em The Dirichlet process, related priors and posterior asymptotics}, page 35–79.
\newblock Cambridge Series in Statistical and Probabilistic Mathematics. Cambridge University Press, 2010.

\bibitem{Hoppe_1984}
F.~M. Hoppe.
\newblock Pólya-like urns and the {E}wenś sampling formula.
\newblock {\em Journal of mathematical biology}, 20(1):91--94, 1984.

\bibitem{Jiang_2009}
W.~Jiang and C.-H. Zhang.
\newblock {General maximum likelihood empirical Bayes estimation of normal means}.
\newblock {\em The Annals of Statistics}, 37(4):1647 -- 1684, 2009.

\bibitem{REBayes_R}
R.~Koenker and J.~Gu.
\newblock {REBayes}: An {R} package for empirical bayes mixture methods.
\newblock {\em Journal of Statistical Software}, 82(8):1--26, 2017.

\bibitem{Koenker_2014}
R.~Koenker and I.~Mizera.
\newblock Convex optimization, shape constraints, compound decisions, and empirical bayes rules.
\newblock {\em Journal of the American Statistical Association}, 109(506):674--685, 2014.

\bibitem{Lijoi2007}
A.~Lijoi, I.~Pruenster, and S.~Walker.
\newblock Bayesian consistency for stationary models.
\newblock {\em Econometric Theory}, 23:749--759, 2007.

\bibitem{Lindsay_1983}
B.~G. Lindsay.
\newblock {The geometry of mixture likelihoods: A general theory}.
\newblock {\em The Annals of Statistics}, 11(1):86 -- 94, 1983.

\bibitem{McLachlan_2014}
G.~J. McLachlan and S.~Rathnayake.
\newblock On the number of components in a {G}aussian mixture model.
\newblock {\em WIREs Data Mining and Knowledge Discovery}, 4(5):341--355, 2014.

\bibitem{Newton_1998}
M.~A. Newton, F.~A. Quintana, and Y.~Zhang.
\newblock {\em Nonparametric Bayes methods using predictive updating}, pages 45--61.
\newblock Springer New York, New York, NY, 1998.

\bibitem{Newton_1999}
M.~A. Newton and Y.~Zhang.
\newblock A recursive algorithm for nonparametric analysis with missing data.
\newblock {\em Biometrika}, 86(1):15--26, 1999.

\bibitem{Robbins_1951}
H.~Robbins and S.~Monro.
\newblock {A stochastic approximation method}.
\newblock {\em The Annals of Mathematical Statistics}, 22(3):400 -- 407, 1951.

\bibitem{Rubin_1981}
D.~B. Rubin.
\newblock {The Bayesian Bootstrap}.
\newblock {\em The Annals of Statistics}, 9(1):130 -- 134, 1981.

\bibitem{Serfling_2009}
R.~Serfling.
\newblock {\em Approximation Theorems of Mathematical Statistics}.
\newblock Wiley Series in Probability and Statistics. Wiley, 2009.

\bibitem{Venables_2002}
W.~N. Venables and B.~D. Ripley.
\newblock {\em Modern Applied Statistics with S}.
\newblock Springer, New York, fourth edition, 2002.
\newblock ISBN 0-387-95457-0.

\bibitem{Vo_2019}
B.~N. Vo, C.~C. Drovandi, and A.~N. Pettitt.
\newblock {Bayesian Parametric Bootstrap for Models with Intractable Likelihoods}.
\newblock {\em Bayesian Analysis}, 14(1):211 -- 234, 2019.

\end{thebibliography}
